%% file: main.tex
\documentclass{jgaa-art}

\usepackage{graphicx}
\usepackage{amsmath}
\usepackage{amssymb}
\usepackage{hyperref}
\usepackage[capitalise]{cleveref} 
\usepackage{todonotes}

\newtheorem{observation}{Observation}
\newtheorem{cor}{Corollary}

\crefformat{footnote}{#2\footnotemark[#1]#3}
\graphicspath{ {./images/} }

\newcommand{\degr}{\operatorname{deg}}
\newcommand{\OO}{\mathcal{O}}
\newcommand{\NP}{\(\mathsf{NP}\)}
\newcommand{\NPC}{\NP-complete}
\newcommand{\FPT}{\(\mathsf{FPT}\)}
\newcommand{\WO}{\(\mathsf{W}[1]\)}	
\newcommand{\nodeD}{{\sc Node Deletion }}
\newcommand{\PAA}{\mathsf{AA}}
\newcommand{\PAS}{\mathsf{AS}}
\newcommand{\PSA}{\mathsf{SA}}
\newcommand{\PSS}{\mathsf{SS}}
\newcommand{\IS}{\mathsf{IS}}
\newcommand{\yes}{\(\mathsf{Yes}\)}
\newcommand{\no}{\(\mathsf{No}\)}
\newcommand{\remove}[1]{}

\begin{document}
\fancyhead[O]{} 

\HeadingAuthor{Eppstein et al.} 
\HeadingTitle{Parameterized Complexity of Finding Subgraphs} 
\title{Parameterized Complexity of Finding Subgraphs with Hereditary Properties on Hereditary Graph Classes}

\author[first]{David Eppstein}{eppstein@uci.edu}
\author[second]{Siddharth Gupta}{siddhart@post.bgu.ac.il}
\author[first]{Elham Havvaei}{ehavvaei@uci.edu}

\affiliation[first]{Department of Computer Science, University of California, Irvine}

\affiliation[second]{Department of Computer Science, Ben-Gurion University of the Negev}


\maketitle

\begin{abstract}
We investigate the parameterized complexity of finding subgraphs with hereditary properties on graphs belonging to a hereditary graph class. Given a graph $G$, a non-trivial hereditary property $\Pi$ and an integer parameter $k$, the general problem $P(G,\Pi,k)$ asks whether there exists $k$ vertices of $G$ that induce a subgraph satisfying property $\Pi$. This problem, $P(G,\Pi,k)$ has been proved to be \NPC\ by Lewis and Yannakakis. The parameterized complexity of this problem is shown to be \WO-complete by Khot and Raman, if $\Pi$ includes all trivial graphs but not all complete graphs and vice versa; and is fixed-parameter tractable (\FPT), otherwise. As the problem is \WO-complete on general graphs when $\Pi$ includes all trivial graphs but not all complete graphs and vice versa, it is natural to further investigate the problem on restricted graph classes. 

Motivated by this line of research, we study the problem on graphs which also belong to a hereditary graph class
and establish a framework which settles the parameterized complexity of the problem for various hereditary graph classes. In particular, we show that:
\begin{itemize}
	\item $P(G,\Pi,k)$ is solvable in polynomial time when the graph $G$ is co-bipartite and $\Pi$ is the property of being planar, bipartite or triangle-free (or vice-versa).
	\item $P(G,\Pi,k)$ is \FPT\ when the graph $G$ is planar, bipartite or triangle-free and $\Pi$ is the property of being planar,  bipartite or triangle-free, or graph $G$ is co-bipartite and $\Pi$ is the property of being co-bipartite.
	\item $P(G,\Pi,k)$ is \WO-complete when the graph $G$ is $C_4$-free, $K_{1,4}$-free or a unit disk graph and $\Pi$ is the property of being either planar or bipartite. 
\end{itemize}
\end{abstract}

\Body 

\input{intro}

\input{prelim}

\input{fpt}
\input{hardness}
\input{conclusion}

\clearpage

\bibliographystyle{abbrvurl}
\bibliography{main}

\end{document}

%% file: intro.tex

\section{Introduction}

In this paper, we study the parameterized complexity of finding $k$-vertex induced subgraphs in a given hereditary class of graphs, within larger graphs belonging to a different hereditary class of graphs. A prototypical instance of the induced subgraph problem is the $k$-clique problem, which asks whether a given graph $G$ has a clique of size $k$. Although $k$-clique is \WO-complete for general graphs~\cite{DBLP:journals/tcs/DowneyF95}, and \NPC~even when the input graph is constrained to be a multiple-interval graph,~\cite{DBLP:conf/soda/ButmanHLR07}, it is fixed-parameter tractable (\FPT)\footnote{For basic notions in parameterized complexity, see \autoref{sec:prelim}.} in this special case~\cite{DBLP:journals/tcs/FellowsHRV09}. This example, of a \WO-complete problem for general graphs which becomes \FPT{} on constrained inputs, motivates us to seek additional examples of this phenomenon, and more broadly to attempt a classification of induced subgraph problems which can determine in many cases whether a constrained induced subgraph problem is tractable or remains hard.

We formalize a graph property as a set $\Pi$ of the graphs that have the property. A property is \emph{nontrivial} if it is neither empty nor contains all the graphs, and more strongly it is \emph{interesting} if infinitely many graphs have the property and infinitely many graphs do not have the property. A nontrivial graph property $\Pi$ is \emph{hereditary} if it is closed under taking induced subgraphs. That is, if $\Pi$ is hereditary and a graph $G$ belongs to $\Pi$, then every induced subgraph of $G$ also belongs to $\Pi$.  Given a hereditary property $\Pi$, let $\overline\Pi$ be the complementary property, the set of graphs which do not belong to $\Pi$. The \emph{forbidden} set ${\cal F}_\Pi$ of $\Pi$ is the set of graphs that are minimal for $\overline\Pi$: they belong to $\overline\Pi$, but all of their proper induced subgraphs belong to $\Pi$. For a hereditary property $\Pi$, a graph $G$ belongs to $\overline\Pi$ if and if $G$ has no induced subgraph in ${\cal F}_\Pi$. Khot and Raman~\cite{DBLP:journals/tcs/KhotR02} studied the parameterized complexity of the following unified formulation of the induced-subgraph problem, without constraints on the input graph: Given a graph $G$, an interesting hereditary property $\Pi$ and a positive integer $k$, the problem $P(G,\Pi,k)$ asks whether there exists an induced subgraph of $G$ of size $k$ that belongs to $\Pi$. They proved a dichotomy theorem for this problem: If $\Pi$ includes all trivial graphs (graphs with no edges) but not all complete graphs, or vice-versa, then the problem is \WO-complete . However, in all remaining cases, the problem is \FPT. 

Our work studies the parameterized complexity of the problem $P(G,\Pi,k)$, in cases for which it is \WO-complete for general graphs, under the constraint that the input graph $G$ belongs to a hereditary graph class $\Pi_G$. (Note that $\Pi_G$ should be a different class than $\Pi$, for otherwise the problem is trivial: just return any $k$-vertex induced subgraph of the input.)  Given a graph $G$, the interesting hereditary properties $\Pi_G$ and $\Pi$, and an integer $k$, we denote our problem by $P(G, \Pi_G, \Pi, k)$.  The main tool that we use for finding efficient algorithms for $P(G, \Pi_G, \Pi, k)$ is Ramsey's theorem, which allows us to prove the existence of either large cliques or large independent sets in arbitrary graphs, allowing some combinations of input graph size and parameter to be answered immediately without performing a search. For the cases where we find hardness results, we do so by reductions from $P(G, \Pi_G, \IS, k)$ to $P(G, \Pi_G, \Pi, k)$, where $\IS$  is the property of being an independent set. We believe our framework has interest in its own right, as a way to settle a wide class of induced-subgraph properties while avoiding the need to develop many tedious hardness proofs for individual problems.

\subsection{Our Contributions} 

We partition interesting hereditary properties into four classes named $\PAA$, $\PAS$,  $\PSA$, and $\PSS$ as follows. A hereditary property $\Pi$ belongs to:

\begin{itemize}
\item $\PAA$, if it includes all complete graphs and all independent sets.
\item $\PAS$, if it includes all complete graphs but excludes some independent sets.
\item $\PSA$, if it excludes some complete graphs but includes all independent sets.
\item $\PSS$, if it excludes some complete graphs as well as some independent sets.
\end{itemize} 

By Ramsey's theorem, an interesting hereditary property cannot belong to $\PSS$. The interesting cases for the problem $P(G, \Pi_G, \Pi, k)$ with respect to $\Pi$ are either $\Pi \in \PSA$ or $\Pi \in \PAS$. In the other two cases, when $\Pi \in \PAA$ or $\Pi \in \PSS$ the problem $P(G, \Pi_G, \Pi, k)$ is known to be fixed-parameter tractable regardless of $\Pi_G$~\cite{DBLP:journals/tcs/KhotR02} . We prove the following results related to the problem $P(G, \Pi_G,$ $ \Pi, k)$, for these interesting cases:

\begin{itemize}
	\item If $\Pi_G \in \PAS$ and $\Pi \in \PSA$ or vice versa, then the problem $P(G,\Pi_G,\Pi,k)$ is solvable in polynomial time (\autoref{thm:AS_SA}). Although the exponent of the polynomial depends in general on $\Pi$, some classes $\Pi_G$ for which subgraph isomorphism is in \FPT{} also have polynomial-time algorithms for $P(G,\Pi_G,\Pi,k)$ whose exponent is fixed independently of $\Pi$ (\autoref{thm:AS_SA_sgi}). The key insight for these problems is that these assumptions cause $\Pi_G\cap\Pi$ to be a finite set, limiting the value of $k$ and making it possible to perform a brute-force search for an induced subgraph while remaining within polynomial time.
	
	A class of problems of this form that have been extensively studied involve finding cliques in sparse graphs or sparse classes such as planar graphs; beyond being polynomial for any fixed hereditary sparse  $ \PAS$ or class of graphs, it is \FPT{} for general graphs when parameterized by degeneracy, a parameter describing the sparsity of the given graph~\cite{EppStrLof-JEA-13}. Another example problem of this type that is covered by this result is finding planar induced subgraphs of co-bipartite graphs; here, $\Pi$ is the property of being planar, in $\PSA$, and $\Pi_G$ is the property of being co-bipartite, in $\PAS$. Similarly, this result covers finding a $k$-vertex bipartite or triangle-free induced subgraph of a co-bipartite graph, or finding a $k$-vertex co-bipartite induced subgraph of a planar, bipartite, or triangle-free graph. 

 	\item If both $\Pi_G$ and $\Pi$ belong either to $\PAS$ or both belong to $\PSA$, then the problem $P(G,\Pi_G,$ $\Pi,k)$ is in \FPT{} (\autoref{thm:both_AS}). The insight that leads to this result is that large-enough graphs in $\Pi_G$ necessarily contain $k$-vertex cliques (for properties in $\PAS$) or independent sets (for properties in $\PSA$), which also belong to $\Pi$. Therefore, the only instances for which a more complicated search is needed are those for which $k$ is large enough relative to $G$ that the existence of a $k$-vertex clique or independent set cannot be guaranteed. For that range of the parameter $k$, the search complexity is in \FPT.
	
	Problems of this type that have been studied previously include finding independent sets in sparse graph families,
	as well as finding planar induced subgraphs of sparse classes of graphs~\cite{BorEppZhu-JGAA-15}.
	Finding a $k$-vertex graph that belongs to one of the four classes of forests, planar graphs, bipartite graphs, or triangle-free graphs, as an induced subgraph of a graph $G$ that belongs to another of these three classes, belongs to the problems of this type.
	
	\item If $\Pi_G \in \PSS$, then the problem $P(G,\Pi_G,\Pi,k)$ is solvable in polynomial~time (\autoref{thm:SS}). This case is trivial: there can be only finitely many graphs in $\Pi_G$ and we can precompute the answers to each one.
	\item In the remaining cases, $\Pi_G \in \PAA$, while $\Pi$ belongs to $\PAS$ or to $\PSA$. These cases include both problems known to be polynomial, such as finding independent sets in various classes of perfect graphs, problems known to be \FPT, including several other cases of independent sets~\cite{DBLP:conf/iwoca/DabrowskiLMR10}, and problems known to be hard for parameterized computation, such as finding independent sets in unit disk graphs~\cite{DBLP:conf/esa/Marx05}. Therefore, we cannot expect definitive results that apply to all cases of this form, as we obtained in the previous cases. Instead, we provide partial results suggesting that in many natural cases the complexity of $P(G,\Pi_G,\Pi,k)$ is controlled by the complexity of the simpler problem of finding independent sets:
\begin{itemize}
	 \item If $\Pi_G$ is closed under duplication of vertices (strong products with complete graphs), and $\Pi$ contains the graphs $n\cdot K_{\chi(\Pi)}$ (disjoint unions of complete graphs with the maximum chromatic number for $\Pi$), then $P(G, \Pi_G, \Pi, k)$ is as hard as $P(G, \Pi_G,\IS, k)$ (\autoref{th:chr}).
	 
	 Families $\Pi_G$ that meet these conditions, for which finding independent sets is \WO-complete, include the property of being a unit disk graph, the property of being  $C_4$-free, and the property of being $K_{1,4}$-free. Families $\Pi$ that meet these conditions include the property of being either planar or  bipartite. Therefore, $P(G, \Pi_G, \Pi, k)$ is also \WO-complete in these families.
	\item If $\Pi_G \in \PAA$ and is closed under joins with disjoint unions of cliques, and if $\Pi$ contains all joins of an independent set with a disjoint union of cliques that have chromatic number at most $\chi(\Pi) - 1$, then $P(G, \Pi_G,\Pi, k)$ is as hard as $P(G,\Pi_G,\IS, k)$ (\autoref{thm:hard_from_join}).
	\end{itemize}
\end{itemize}

\subsection{Other Related Work}

Before the investigation of the parameterized complexity of $P(G, \Pi,k)$\remove{finding subgraphs with hereditary properties on general graphs}, Lewis and Yannakakis had studied the dual of this problem, the \nodeD problem, for interesting hereditary properties, which is defined as follows: Given a graph $G$ and an interesting hereditary property $\Pi$, find the minimum number of nodes to delete from $G$ such that the resulting graph belongs to $\Pi$. They proved that the \nodeD problem is \NPC~\cite{DBLP:journals/jcss/LewisY80}. Cai~\cite{DBLP:journals/ipl/Cai96} studied the parameterized version of \nodeD and proved that the problem is \FPT, parameterized by the number of deleted vertices, for an interesting hereditary property with a finite forbidden set.

Related to our line of work on the parameterized complexity of hereditary properties, finding an independent set with the maximum cardinality (MIS) on a general graph, has been proved to be NP-hard even for planar graphs of degree at most three~\cite{DBLP:books/fm/GareyJ79}, unit disk graphs~\cite{DBLP:journals/dm/ClarkCJ90}, and $C_4$-free graphs~\cite{alekseev1982effect}. Fellows et al. proved that finding a $k$-Independent Set is \WO-hard for 2-interval graphs while its complementary problem, $k$-clique, as mentioned before is \FPT\  for multiple-interval graphs~\cite{DBLP:journals/tcs/FellowsHRV09}.

%% file: prelim.tex

\section{Preliminaries}\label{sec:prelim}

Throughout the paper, we consider finite undirected graphs. Given a graph $G$, we denote its vertex set and edge set by $V(G)$ and $E(G)$, respectively. For a vertex $v\in V(G)$, we denote the set of all adjacent vertices of $v$ in $G$ by $N_G(v)$, i.e. $N_G(v)=\{u\in V(G)~|~\{u,v\}\in E(G)\}$. The degree of a vertex $v\in V(G)$ in $G$ is denoted by $\degr_G(v)$. Given a vertex set $S \subseteq V(G)$, $G[S]$ represents the subgraph of $G$ induced by $S$. The \emph{chromatic number} of a graph $G$ is the minimum number of colors needed to color the vertices such that no two adjacent vertices get the same~color. The chromatic number of a graph property $\Pi$ is the maximum chromatic number of any graph $G \in \Pi$.

Let $\Pi$ be a hereditary graph property. If $\Pi \in \PAS$ or $\Pi \in \PSS$, then we denote the size of the smallest independent set that does not belong to $\Pi$ by $i_\Pi$. Similarly, if $\Pi \in \PSA$ or $\Pi \in \PSS$, then we denote the number of vertices in the smallest clique that does not belong to $\Pi$ by $c_\Pi$. We denote the property of being an independent set (the family of all all independent sets) as $\IS$.
 %


 The use of parameterized complexity has been growing remarkably, in recent decades. What has emerged is a very extensive collection of techniques in diverse areas on numerous parameters.  A problem $L$ is a {\em parameterized} problem if each problem instance of $L$ is associated with a {\em parameter} $k$. For simplicity, we represent an instance of a parameterized problem $L$ as a pair $(I,k)$ where $k$ is the parameter associated with input $I$. Formally, we say that $L$ is {\em fixed-parameter tractable} if any instance $(I, k)$ of $L$ is solvable in time $f(k)\cdot |I|^{\OO(1)}$, where $|I|$ is the number of bits required to specify input $I$ and $f$ is a computable function of $k$. We remark that this framework also provides methods to show that a parameterized problem is unlikely to be \FPT. The main technique is the one of parameterized reductions analogous to those employed in classical complexity, with the concept of \WO-hardness replacing  \NP-hardness. For problems whose solution is a set (for instance of vertices or edges), the size of this set is a natural parameter for the study of the parameterized complexity of the problem. Various problems such as  $k$-vertex cover~\cite{DBLP:journals/siamcomp/BussG93, DBLP:conf/mfcs/ChenKX06, DBLP:journals/jal/ChenKJ01}, $k$-feedback vertex set~\cite{DBLP:conf/stoc/ChenLL08} have been studied under this definition of natural parameter. There are numerous examples of other studies not solely parameterized by the size of the solution~\cite{DBLP:conf/sat/Szeider03, DBLP:journals/jgaa/BannisterE18, DBLP:journals/algorithmica/EppsteinH20}. In this paper, we study our problems under their natural parameter, the number of vertices of the subgraph we are seeking. For more information on parameterized complexity, we refer the reader to~\cite{DBLP:books/sp/CyganFKLMPPS15,DBLP:series/txcs/DowneyF13}. 
%

%% file: fpt.tex

\section{Tractability Results}


In this section, we identify pairs of hereditary properties $\Pi_G$ and $\Pi$ for which the problem $P(G,\Pi_G,\Pi, k)$ is either in \textsf{P} or \FPT.  Our proofs  use Ramsey numbers which we begin by defining. For any positive integers $r$ and $s$, there exists a minimum positive integer $R(r,s)$ such that any graph on at least $R(r,s)$ vertices contains either a clique of size $r$ or an independent set of size $s$. It is well-known that $R(r,s) \leq \tbinom{ r + s - 2 }{r - 1}$~\cite{DBLP:books/daglib/0070576}. It will also be convenient in our analysis to have a notation for the time to test whether a given $k$-vertex graph (typically, a subgraph of our given graph $G$) has property $\Pi$; we let $t_\Pi(k)$ denote this time complexity.

\begin{theorem}\label{thm:AS_SA}
If $\Pi_G \in \PAS$ and $\Pi \in \PSA$ or vice versa, then the problem $P(G,\Pi_G,$ $\Pi,k)$ is solvable in polynomial time.
\end{theorem}

\begin{proof}
We give a proof for the case when $\Pi_G \in \PAS$ and $\Pi \in \PSA$. The proof for the other case is symmetric under reversal of the roles of cliques and independent sets. Recall that every graph on $R(c_\Pi,i_{\Pi_G})$ vertices contains either a clique of size $c_\Pi$, too large to have property $\Pi$, or it contains an independent set of size $i_{\Pi_G}$, too large to have property $\Pi_G$.
Therefore, If $k \geq R(c_\Pi,i_{\Pi_G})$, it is impossible for a $k$-vertex induced subgraph of a graph $G$ in $\Pi_G$
to also have property $\Pi$, because such a subgraph would either have a large clique (contradicting the membership of the subgraph in $\Pi$) or a large independent set (contradicting the membership of $G$ in $\Pi_G$). Therefore, for such large values of $k$, an algorithm for $P(G,\Pi_G,$ $\Pi,k)$ can simply answer \no{} without doing any searching.

If $k < R(c_\Pi,i_{\Pi_G})$, then we can use a brute force search to test whether there exists a $k$-vertex induced subgraph having property $\Pi$. Specifically, we enumerate all $k$-vertex subsets of the vertices of $G$, construct the induced subgraph for each subset, and test whether any of these induced subgraphs belongs to $\Pi$. Given a representation of $G$ for which we can test adjacency in constant time, the time to construct each subgraph is $O(k^2)$, so the total time taken by this search is
\[
\binom{n}{k} \left( O(k^2)+ t_\Pi(k) \right) \le n^r \left( O(r^2)+ t_\Pi(r) \right),
\]
where $r=R(c_\Pi,i_{\Pi_G})-1$. As the right hand side of this time bound is a polynomial of $n$ without any dependence on $k$, this is a polynomial time algorithm. Thus, the problem $P(G,\Pi_G,\Pi,k)$ is solvable in polynomial time.
\end{proof}

Although polynomial, the time bound of \autoref{thm:AS_SA} has an exponent $r$ that depends on $\Pi$ and $\Pi_G$, and may be large.
An alternative approach, which we outline next, may lead to better algorithms for properties $\Pi_G$ for which the induced subgraph isomorphism problem is in \FPT, as it is for instance for planar graphs~\cite{DBLP:journals/jgaa/Eppstein99} or more generally for nowhere-dense families of graphs~\cite{NesOss-SGI-12}.

\begin{theorem}\label{thm:AS_SA_sgi}
If $\Pi_G \in \PAS$ and $\Pi \in \PSA$ or vice versa, and induced subgraph isomorphism is in \FPT{} in $\Pi_G$ with time $t_{\operatorname{sgi}}(n,k)$ to find $k$-vertex induced subgraphs of $n$-vertex graphs,
then the problem $P(G,\Pi_G,$ $\Pi,k)$ is solvable in polynomial time $O(t_{\operatorname{sgi}}(n,r))$,
for the same constant $r$ (depending on $\Pi$ and $\Pi_G$ but not on $k$ or $G$) as in \autoref{thm:AS_SA}.
\end{theorem}

\begin{proof}
If $k > r$, we answer \no{} immediately as in \autoref{thm:AS_SA}.
Otherwise, we generate all $k$-vertex graphs, test each of them for having property $\Pi$,
and if so apply the subgraph isomorphism algorithm for graphs with property $\Pi_G$ to $G$ and the generated graph.
There are $2^{O(r^2)}$ graphs to generate, testing for property $\Pi$ takes time $t_\Pi(r)$ for each one,
and testing for being an induced subgraph of $G$ takes time $t_{\operatorname{sgi}}(n,r)$ for each one, so the time is as stated.
\end{proof}

In particular, these problems can be solved in linear time for planar graphs.

\begin{theorem}\label{thm:both_AS}
If both $\Pi_G$ and $\Pi$ belong to $\PAS$, or if both belong to $\PSA$, then the problem $P(G,\Pi_G,$ $\Pi,k)$ is in \FPT.
\end{theorem}

\begin{proof}
We give a proof for the case when both $\Pi_G$ and $\Pi$ belong to $\PAS$. The proof for the other case is again symmetric under reversal of the roles of cliques and independent sets. For a graph $G\in\Pi_G$ that is large enough that $|V(G)| \geq R(k,i_{\Pi_G})$, it must be the case that $G$ contains a clique $C$ of size $k$, for it cannot contain an independent set of size $i_{\Pi_G}$ without violating the assumption that it belongs to $\Pi_G$. Because $\Pi$ is assumed to be in $\PAS$, it contains all cliques, so this $k$-vertex clique belongs to $\Pi$. Therefore, for graphs with this many vertices, it is safe to answer \yes.
There is a small subtlety here, in that we do not know an efficient method to calculate $R(k,i_{\Pi_G})$, and an inefficient method would unnecessarily increase the dependence of our time bounds on the parameter~$k$. However, we can use the inequality
\[
R(k,i_{\Pi_G}) \le \binom{k+i_{\Pi_G}-2}{k-1}
\]
to get a bound on this number that is easier to calculate.
Our algorithm can simply test whether
$|V|\ge\tbinom{k+i_{\Pi_G}-2}{k-1}$, and if so we return \yes{} without doing any searching. 

If $|V(G)| < \tbinom{k+i_{\Pi_G}-2}{k-1}$, then constructing and checking all induced subgraphs of $G$ of size $k$ to detect whether there exists such a subgraph belonging to $\Pi$ takes time
\[
\binom{k+i_{\Pi_G}-2}{k-1}^k \left(O(k^2) + t_\Pi(k)\right),
\]
a time complexity that is bounded by a function of $k$ but independent of $n$.
As the times for both cases are of the appropriate form, the problem $P(G,\Pi_G,\Pi,k)$ is in \FPT.
\end{proof}

The following corollaries can be directly obtained from \autoref{thm:AS_SA} and \autoref{thm:both_AS}.

\begin{cor}\label{co:poly}
If $\Pi_G$ is the property of being co-bipartite and $\Pi$ is the property of being a forest, planar, bipartite or triangle-free (or vice versa), then the problem $P(G,\Pi_G,\Pi,k)$ is solvable in polynomial time.
\end{cor}

\begin{cor}\label{co:FPT}
If $\Pi_G$ and $\Pi$ are the properties of being planar, bipartite or triangle-free, then the problem $P(G,\Pi_G,\Pi,k)$ is \FPT.
\end{cor}

For completeness, we state the following (trivial) theorem:

\begin{theorem}\label{thm:SS}
If $\Pi_G \in \PSS$, then the problem $P(G,\Pi_G,\Pi,k)$ is solvable in polynomial time.
\end{theorem}

\begin{proof}
We have  $|V(G)| < R(c_{\Pi_G},i_{\Pi_G})$, because otherwise $G$ has either a clique of size $c_{\Pi_G}$ or a trivial graph of size $i_{\Pi_G}$, a contradiction. Because $V(G)$ is bounded, there are only finitely many valid inputs to the problem $P(G,\Pi_G,\Pi,k)$ and we can precompute the solutions to each one.
\end{proof}

\autoref{tb:fpt} briefly summarizes the results of Theorems \ref{thm:AS_SA}, \ref{thm:both_AS} and \ref{thm:SS}. 

\begin{table}[!t]
\centering
\begin{tabular}{ r|c|c| }
\multicolumn{1}{r}{}
 &  \multicolumn{1}{c}{$\Pi \in \PSA$}
 & \multicolumn{1}{c}{$\Pi \in \PAS$} \\
\cline{2-3}
$\Pi_G \in  \PAS$ & {\shortstack{\\ \\ If $ k < R(c_\Pi,i_{\Pi_G})$ \\ check all induced subgraphs \\ of size $k$, otherwise return \no}}  &  {\shortstack{If $|V(G)| < \tbinom{k+i_{\Pi_G}-2}{k-1}$ \\ check all induced subgraphs \\ of size $k$, otherwise return \yes}} \\[0.3cm]
\cline{2-3}
$\Pi_G \in \PSA$ & {\shortstack{\\ \\If $|V(G)| < \tbinom{k+c_{\Pi_G}-2}{k-1}$ \\check all induced subgraphs \\ of size $k$, otherwise return \yes}}  & 
 {\shortstack{If $ k < R(c_{\Pi_G},i_\Pi)$ \\check all induced subgraphs\\ of size $k$, otherwise return \no}} \\[0.3cm]
\cline{2-3}
$\Pi_G \in  \PSS$ & \multicolumn{2}{|c|}{\shortstack{\\ \\$|V(G)| < R(c_{\Pi_G},i_{\Pi_G})$, precompute all possible inputs}} \\ [0.3cm]\cline{2-3}
\end{tabular}

\caption{Summary of Theorems \ref{thm:AS_SA}, \ref{thm:both_AS} and \ref{thm:SS}.  \label{tb:fpt}}

\end{table}

%% file: hardness.tex

\section{Hardness from strong products}\label{sec:hardness}

In this section, we prove some hardness results for the problem $P(G,\Pi_G,\Pi,k)$, when $\Pi_G \in \PAA$ and $\Pi \in \PSA$. 

\subsection{Hardness from strong products with cliques}
To formulate the first of these results in full generality, we need some definitions. The \emph{strong product} $G\boxtimes H$ is defined as a graph whose vertex set $V(G)\times V(H)$ consists of the ordered pairs of a vertex in $G$ and a vertex in $H$, with two of these ordered pairs $(u,v)$ and $(u',v')$ adjacent if $u$ and $u'$ are adjacent or equal, and $v$ and $v'$ are adjacent or equal. In particular, the strong product with a complete graph, $G\boxtimes K_i$, can be thought of as making $i$ copies of each vertex in $G$, with two copies of the same vertex always adjacent, and with adjacency between copies of different vertices remaining the same as in $G$. We use the notation $n\cdot K_i$ to denote the disjoint union of $n$ copies of an $i$-vertex complete graph; this is the strong product of an $n$-vertex independent set with an $i$-vertex clique. Given a family of graphs $\Pi$, we let $\chi(\Pi)$ denote the maximum chromatic number (if it exists) of the graphs in $\Pi$.

\begin{observation}\label{obs:ISsize}
Let $G$ be a graph on $n$ vertices with chromatic number $\chi(G)$. Then, there exists an independent set of $G$ of size at least $n/\chi(G)$.
\end{observation}

Namely, the large independent set of the observation can be chosen as the largest color class of any optimal coloring of~$G$.


\begin{theorem}\label{th:chr}
Let $\Pi_G \in \PAA$ be a hereditary property which is closed under strong products with complete graphs, and let $\Pi \in \PSA$ be a hereditary property such that, for all $n$, the graph $n\cdot K_{\chi(\Pi)}$ belongs to $\Pi$. Then, the problem $P(G,\Pi_G,\Pi,k)$ is as hard as $P(G,\Pi_G,\IS, k)$.
\end{theorem}

\begin{proof}
We describe a polynomial-time parameterized reduction from instances of $P(G,\Pi_G,\IS, k)$ to equivalent instance of $P(G,\Pi_G,\Pi,k')$, where $k'$ depends only on $k$ (and not on $G$). The reduction transforms the graph $G$ of the instance into a new graph $G'=G\boxtimes K_{\chi(G)}$, and transforms the parameter $k$ into a new parameter value $k'=k\cdot\chi(G)$.

As we have assumed that $\Pi_G$ is closed under strong products with complete graphs, it follows that $G' \in \Pi_G$, so the reduction produces a valid instance of $P(G,\Pi_G,\Pi,k')$. To show that this instance is equivalent to the starting instance, we show that $G$ has an independent set of size $k$ if and only if $G'$ has an induced subgraph of size $k'$  belonging to $\Pi$.
\begin{description}
\item[($\Rightarrow$)] Let $I$ be an independent set of $G$ of size $k$, and let $X=I\boxtimes K_{\chi(G)}$ be the subgraph of $G'$ induced by the set of all copies of vertices in $I$. Then $|V(X)|=k'$ and, as a graph of the form $k\cdot K_{\chi(G)}$, $X$ belongs to $\Pi$ by assumption.

\item[($\Leftarrow$)] Let $H \in \Pi$ be an induced subgraph of $G'$ of size $k'$. By \autoref{obs:ISsize}, it has an independent set $I'$ of size $k'/\chi(G)\ge k$. This independent set can include at most one copy of each vertex in $G$, so the set $I$ of vertices in $G$ whose copies are used in $I'$ must also have size $\ge k$. Further, $I$ is independent, for any edge between its vertices would be copied as an edge in $G'$, contradicting the assumption that we have an independent set in $G'$. Therefore, $I$ is an independent set of size $\ge k$ in $G$, as desired.
\end{description}
\end{proof}

The families of unit-disk graphs, $C_4$-free graphs, and $K_{1,4}$-free graphs all belong to $\PAA$, and are closed under strong products with complete graphs. Finding independent sets is also known to be complete for unit-disk graphs~\cite{DBLP:conf/esa/Marx05,DBLP:conf/iwpec/Marx06}, $C_4$-free graphs~\cite{DBLP:conf/iwpec/BonnetBCTW18}, and $K_{1,4}$-free graphs~\cite{DBLP:journals/algorithmica/HermelinML14}. Moreover, the families of planar graphs and of bipartite graphs both have the property that $n\cdot K_{\chi(\Pi)}\in\Pi$. For instance, in planar graphs, the graph $n\cdot K_{\chi(\Pi)}$ consists of $n$ disjoint copies of $K_4$, a planar graph, and forming disjoint unions preserves planarity. Therefore, we have the following corollary:

\begin{cor}
\label{chr-example}
If $\Pi_G$ is the property of being $(a)$ unit-disk, $(b)$ $C_4$-free, or $(c)$ $K_{1,4}$-free , and $\Pi$ is the property of being either  planar  or  bipartite, then the problem $P(G,\Pi_G,\Pi,k)$ is \WO-complete. 
\end{cor}

\subsection{Hardness from joins with cliques}

The \emph{join} of two graphs $G+H$ is a graph formed from the disjoint union of $G$ and $H$ by adding edges from each vertex of $G$ to each vertex of $H$. The reduction that we consider in this section involves the join with a disjoint union of cliques,
$G+t\cdot K_c$. That is, starting from $G$ we add $t$ cliques of size $c$, with each vertex in $G$ connected to all vertices in these cliques.

\begin{observation}\label{rem:maxCliqueSize}
Let $G$ have maximum clique size $\omega(G)$, and let $t$ and $c$ be positive integers.
Then the maximum clique size of $G+t\cdot K_c$ is $\omega(G)+c$.
\end{observation}

\begin{theorem}\label{thm:hard_from_join}
Let $\Pi_G \in \PAA$ be a hereditary property which is closed under joins with disjoint unions of cliques, and $\Pi \in \PSA$ be a hereditary property which includes all subgraphs $I+n\cdot K_{\chi(\Pi)-1}$ for an independent set $I$ and positive integer $n$.
Then the problem $P(G,\Pi_G,\Pi,k)$ is as hard as $P(G,\Pi_G,\IS,k)$. 
\end{theorem}

\begin{proof}
We first construct a new graph $G'=G+r\cdot K_c$,
where $r=R(\chi(\Pi)+1,k)$ and $c=\chi(\Pi)-1$, and a new parameter value $k'=k+rc$.
By the assumption that $\Pi_G$ is closed under joins with disjoint unions of cliques, $G'\in\Pi_G$. Now, we show that $G$ has an independent set of size $k$ if and only if $G'$ has an induced subgraph of size $k'$ belonging to $\Pi$.

\begin{description}
\item[($\Rightarrow$)] Let $I$ be an independent set of $G$ of size $k$. Consider the induced subgraph $I+r\cdot K_c$ of $G'$,
formed by including all vertices that were added to $G$. This subgraph has size $k'=k+rc$, and by assumption it belongs to $\Pi$.
 
\item[($\Leftarrow$)] Let $H \in \Pi$ be an induced subgraph of $G'$ of size $k'$. The vertices of $H$ can be partitioned into two sets $S_1\subset V(G)$ and $S_2\subset r\cdot K_c$. The following two cases can occur:

\begin{itemize}
\item If $S_1$ is not an independent set, let $uv$ be an edge in $S_1$. Then $S_2$ must have at most $c-1$ vertices in each clique of $r\cdot K_c$, for if it contained all $c$ vertices of one of these cliques, then these $c$ vertices together with $u$ and $v$ would form a clique of size $\chi(\Pi)+1$, which is disallowed in $\Pi$. Therefore, $S_2$ has at most $r(c-1)$ vertices, and to obtain total size $k'$, $S_1$ must have at least $k+r$ vertices. By the definition of $r$ and by Ramsey's theorem, $S_1$ has either a clique of size $\chi(\Pi)+1$ (again, an impossibility) or an independent set of size $k$, as desired.

\item If $S_1$ is  an independent set, we observe that, even if $S_2$ includes all of the vertices added to $G$ to form $G'$, it has only $rc$ vertices. Therefore, to obtain total size $k'$, $S_1$ must have at least $k$ vertices, and contains an independent set of size $k$, as desired.
\end{itemize}
\end{description}
\end{proof}

There are many families $\Pi_G$ that meet the requirements on $\Pi_G$ in this theorem, but do not meet the requirements of \autoref{th:chr}: this will be true, for instance, when the forbidden subgraphs of $\Pi_G$ do not include disjoint unions of cliques, and are co-connected (so they cannot be formed by joins, which produce co-disconnected graphs) but at least one of these graphs contains two adjacent twin vertices (with the same neighbors other than each other). The requirement on $\Pi$ in this theorem is met, for instance, by the family $\Pi$ of bipartite graphs. In this case, $\chi(\Pi)=2$, so the graphs $I+n\cdot K_{\chi(\Pi)-1}$ are just  complete bipartite graphs, which are of course bipartite.

As an example, finding $k$-independent sets in $\overline{K_{1,3}}$-free graphs (the complements of claw-free graphs) is known to be NP-complete, from the completeness of the same problem in triangle-free graphs~\cite{MR351881}. \autoref{thm:hard_from_join} then shows that finding $k$-vertex bipartite induced  subgraphs of $\overline{K_{1,3}}$-free graphs is also NP-complete. However, we cannot use this method to prove parameterized hardness for this example, because the $k$-independent set problem in $\overline{K_{1,3}}$-free graphs can be solved in \FPT{} by applying an \FPT{} algorithm for $(k-1)$-independent sets in triangle-free graphs~\cite{DBLP:conf/iwoca/DabrowskiLMR10} to the sets of non-neighbors of all vertices.

%% file: conclusion.tex

\section{Conclusion}\label{sec:concl}
 
We have further narrowed down the parameterized complexity of the problem $P(G,\Pi,k)$ for the case when it is \WO-complete. In particular, restricting the input graph $G$ to belong to a hereditary graph class $\Pi_G$ helps us to settle parameterized complexity of numerous graph classes circumventing long and tedious reduction proofs.\remove{We have shown the following results:
 
 \begin{itemize}
 \item If $\Pi_G \in \PAS$ and $\Pi \in \PSA$ or vice versa, then the problem $P(G,\Pi_G,\Pi,k)$ is solvable in polynomial time.

 \item  If both $\Pi_G, \Pi \in \PAS$ or both $\Pi_G, \Pi \in \PSA$, then the problem is \FPT. 
 
 \item If $\Pi_G \in \PSS$, then the problem $P(G,\Pi_G,\Pi,k)$ is solvable in polynomial time.
 
 \item If $\Pi_G \in \PAA$ be a hereditary property which is closed under modular operation, and $\Pi \in \PSA$ be a hereditary property for which the chromatic number is less than or equal to the maximum allowed clique size $c_{\Pi}-1$. Then, the parameterized complexity of the problem is as hard as finding an independent set of size $k$ in $G$.
 
 \item If $\Pi_G \in \PAA$ be a hereditary property which is closed under mega biclique operation, and $\Pi \in \PSA$ be a hereditary property which does not have generalized bicliques, of maximum clique size $c_\Pi - 1$, as forbidden subgraphs. Then, the parameterized complexity of the problem is as hard as finding an independent set of size $k$ in $G$.
 
\end{itemize} 
 
These characterizations contribute to the parameterized complexity of the problem; when both $\Pi$ and $\Pi_G$ are the properties of being planar, bipartite, triangle-free, or co-bipartite, then the problem is \FPT. Additionally, if $\Pi_G$ is the property of being  unit-disk graph, $C_4$-free, or $K_{1,4}$-free, and $\Pi$ is the property of being either  planar or  bipartite, then the problem $P(G,\Pi_G,\Pi,k)$ is \WO-complete.} 
It remains an open problem to determine the parameterized complexity of the problem $P(G,\Pi_G,\Pi,k)$ when $\Pi_G \in \PAA$ without any restrictions. It would be also interesting to investigate this problem under other graph parameters beyond the size of the solution.